\documentclass[11pt]{article}
\usepackage[margin=1in]{geometry}

\usepackage[utf8]{inputenc}					
\usepackage{csquotes}
\usepackage{amsfonts,amsmath,amssymb,dsfont,mathtools}
\usepackage[hypertexnames=false]{hyperref}
\usepackage{amsthm,comment}
\usepackage[inline, shortlabels]{enumitem}

\usepackage[dvipsnames]{xcolor}

\usepackage{tikz}
\usepackage{appendix}

\date{}

\usepackage{bbold}
\usepackage{mathtools}

\usepackage{algorithm2e}

\usepackage{hyperref}						
\hypersetup{
	hidelinks,
    colorlinks = true,
    linkcolor = MidnightBlue,
    citecolor = MidnightBlue
}

\usepackage[noabbrev,capitalize,nameinlink]{cleveref}			

\makeatletter
\newtheorem*{rep@theorem}{\rep@title}
\newcommand{\newreptheorem}[2]{%
\newenvironment{rep#1}[1]{%
 \def\rep@title{#2 \ref{##1} (restated)}%
 \begin{rep@theorem}}%
 {\end{rep@theorem}}}
\makeatother

\newtheorem{theorem}{Theorem}[section]
\newtheorem{lemma}[theorem]{Lemma}
\newtheorem{claim}{Claim}[theorem]

\newtheorem{observation}[theorem]{Observation}
\theoremstyle{definition}
\newtheorem{definition}[theorem]{Definition}

\crefname{claim}{Claim}{Claims}
\crefname{observation}{Observation}{Observations}
\crefname{equation}{Eq.}{Eqs.}

\newcommand{\triangleqed}{\renewcommand{\qedsymbol}{\ensuremath{\triangle}}}

\DeclareMathOperator{\bbN}{\mathbb{N}}

\DeclareMathOperator{\cO}{\mathcal{O}}

\newcommand{\argmax}{\textrm{argmax}}

\renewcommand{\Pr}{\mathbb{P}}
\newcommand{\E}{{\rm I\kern-.3em E}}

\newcommand{\bit}{\mathrm{bit}}
\newcommand{\carry}{\mathrm{carry}}
\newcommand{\States}{\mathcal{S}}
\newcommand{\State}[2][]{#2.\mathrm{state}_{#1}}

\newcommand{\subc}{\texttt{Succ}}

\newcommand{\Next}{\texttt{Next}}

\newcommand{\Value}{\mathrm{Value}}
\newcommand{\Trains}{\mathrm{Trains}}
\newcommand{\Add}{\texttt{Add}}
\newcommand{\WagonUpdate}{\texttt{Wagon-Update}}

\newcommand{\NewLeader}{\texttt{New-Leader}}

\newcommand{\WagonUpdateLeader}{\texttt{Wagon-Creation}}
\newcommand{\IsKilled}{\texttt{Is-Eliminated}}

\newcommand{\id}{\mathrm{idx}}

\newcommand{\leader}{\mathrm{leader}}

\newcommand{\flag}{\mathrm{flag}}
\newcommand{\rand}{\mathrm{rand}}

\newcommand{\SuccIsMarked}{\texttt{SuccIsMarked}}

\newcommand{\err}{\texttt{Err}}
\newcommand{\ErrSuccessor}{\texttt{Err-Successor}}
\newcommand{\ErrOverflow}{\texttt{Err-Overflow}}

\newcommand{\legitimate}{\Gamma_{L}} 
\newcommand{\execution}{(\gamma_t)_{t \in \bbN}}

\begin{document}

\title{Informative Trains: A Memory-Efficient Journey to a Self-Stabilizing Leader Election Algorithm in Anonymous Graphs}

\author{Lelia Blin\thanks{Université Paris Cité, CNRS, IRIF, Paris, France. E-mail: \texttt{lelia.blin@irif.fr}} \and Sylvain Gay\thanks{IRIF and École Normale Supérieure, Paris, France. E-mail: \texttt{sylvain.gay@ens.psl.eu}}
 \and Isabella Ziccardi\thanks{CNRS, Université Paris Cité, IRIF, Paris, France. E-Mail: \texttt{isabella.ziccardi@irif.fr}. Research supported in
part by the French PEPR integrated project EPiQ (ANR-22-PETQ-0007).}
}

\maketitle

\begin{abstract}
We study the self-stabilizing leader election problem in anonymous $n$-nodes networks. Achieving self-stabilization with low space memory complexity is particularly challenging, and designing space-optimal leader election algorithms remains an open problem for general graphs. In deterministic settings, it is known that $\Omega(\log \log n)$ bits of memory per node are necessary [Blin et al., Disc. Math. \& Theor. Comput. Sci., 2023], while in probabilistic settings the same lower bound holds for some values of \(n\), but only for an unfair scheduler [Beauquier et al., PODC 1999]. 
Several deterministic and probabilistic protocols have been proposed in models ranging from the state model to the population protocols.
However, to the best of our knowledge, existing solutions either require $\Omega(\log n)$ bits of memory per node for general worst case graphs, or achieve low state complexity only under restricted network topologies such as rings, trees, or bounded-degree graphs.

In this paper, we present a probabilistic self-stabilizing leader election algorithm for arbitrary anonymous networks that uses $\cO(\log \log n)$ bits of memory  per node.
Our algorithm operates in the state model under a synchronous scheduler and assumes knowledge of a global parameter $N = \Theta(\log n)$. We show that, under our protocol, the system converges almost surely to a stable configuration with a unique leader and stabilizes within $\cO(\mathrm{poly}(n))$ rounds with high probability.
To achieve $\cO(\log \log n)$ bits of memory, our algorithm keeps transmitting information after convergence, i.e. it does not verify the silence property. Moreover, like most works in the field, our algorithm does not provide explicit termination detection (i.e., nodes do not detect when the algorithm has converged).

Our algorithm is inspired by the uniform constant-state leader election protocol in [Vacus and Ziccardi, PODC 2025] for the beeping model (without self-stabilization). A novel component of our approach is the introduction of an information train mechanism that propagates $\cO(\log n)$-bit information across the network while requiring only $\cO(\log \log n)$ local memory per node. We believe this mechanism may be of independent interest for the design of space-efficient self-stabilizing distributed algorithms.
\end{abstract}

\section{Introduction}

Self-stabilization \cite{Dijkstra74}
represents a comprehensive framework for endowing networks with autonomous recovery properties after transient faults. Formally, a self-stabilizing protocol guarantees convergence to correct behavior from any arbitrary initial state, without requiring external intervention or reinitialization. Among the core primitives of distributed systems, leader election holds particular significance, as it establishes a distinguished process capable of coordinating system-wide activities. The interplay between self-stabilization and leader election is especially critical: numerous distributed protocols presuppose the existence of a unique leader throughout the execution, including in the presence of transient failures. Consequently, a self-stabilizing leader election protocol enables the deployment of such algorithms in systems lacking a predetermined coordinator, leveraging well-established composition techniques that preserve stabilization properties.

Self-stabilizing leader election (SSLE) is a fundamental problem in distributed computing, and has been extensively studied in both anonymous and ID-based networks. 
In anonymous networks, randomization is essential, as deterministic symmetry breaking is impossible~\cite{angluin1980local}. 

Two key performance measures are typically considered in the design of self-stabilizing algorithms: the memory used by each node, and the time required to converge to a stable configuration. 
Memory efficiency directly impacts the amount of information that nodes need to exchange with their neighbors during stabilization. 
Algorithms with smaller state spaces not only reduce communication overhead in fault-free executions or after stabilization, but also facilitate the integration of self-stabilization with other distributed functionalities such as replication.

Designing SSLE algorithms with optimal space complexity is particularly challenging. 
These algorithms have been studied under various distributed computing models, including population protocols~\cite{chen2020self,AustinBFGH25,gasieniec2025improving}, stone-age models~\cite{EmekK21}, and the classical state model~\cite{BeauquierGJ07,Itkis1995fast,BlinT20}. 
However, achieving optimal space complexity for arbitrary graphs remains an open problem in both deterministic and randomized settings, for both anonymous networks and networks with unique IDs. In fact, all known algorithms in general $n$-nodes graphs require $\Omega(\log n)$ bits of memory per node in worst-case topologies \cite{AustinBFGH25,EmekK21,BlinT20}.
To date, space-optimal solutions are known only for specific graph families, such as rings~\cite{BeauquierGJ07}, bounded-degree graphs~\cite{chen2020self}, and bounded-diameter graphs~\cite{EmekK21}.

From the lower-bound perspective, it is known that any deterministic leader election algorithm requires at least $\Omega(\log \log n)$ bits of memory per node~\cite{BlinFB23} in an $n$-node ring. 
Regarding probabilistic algorithms, in~\cite{BeauquierGJ99} they prove a lower bound of $\Omega(\log \log n)$ in rings, for some values of $n$ with an unfair scheduler.

In this paper, we present the first self-stabilizing leader election algorithm for general connected $n$-node graphs that uses sublogarithmic memory. 
Our algorithm is randomized and works in anonymous networks with a synchronized scheduler. It can be implemented in the state model,  requires only $\cO(\log \log n)$ bits of memory per node for \emph{any} graph and converges in polynomial time in $n$. In particular, we assume that nodes have knowledge of a common parameter $N = \Theta(\log n)$, and when this approximation of $n$ is sufficiently accurate, we can guarantee convergence in $\cO(n^3 \log n)$ rounds. 
The algorithm is non-silent: nodes continue to change their states after convergence. This is necessary to achieve $\cO(\log \log n)$ memory, as silent self-stabilizing leader election requires at least $\Omega(\log n)$ bits per node~\cite{DolevGS99}. We do not claim space optimality, since the only known lower bound of $\Omega(\log \log n)$ bits applies to deterministic algorithms~\cite{BlinFB23} or to randomized algorithms with unfair daemons~\cite{BeauquierGJ99}.

\subsection{Model and Problem Definition}

In this section, we define the distributed model considered in this paper and formalize the notion of self-stabilization. 

\paragraph{The Model} A \textit{distributed system} is modeled as an undirected connected graph, $G=(V,E)$, where $V$ is a set of $n$ nodes, and $E$ is the set of edges connecting them. Nodes represent processes, and edges represent bidirectional communication links, allowing two adjacent nodes to communicate directly.

We consider the \emph{state model} \cite{Dijkstra74} in the context of \emph{uniform anonymous networks}. 
In this model, the communication paradigm allows each node to read in an atomic step its own state as well as the state of its neighbors, and to write its own state.
A uniform anonymous network is characterized by the absence of node identifiers and the absence of any distinguished node. We deal with a uniform protocol, meaning all nodes adhere to the same protocol.

Traditionally, algorithms in the state model are described using a set of rules, which are applied to a node's state under certain conditions. We use a different but equivalent formalism. 
A protocol is defined by a finite set of \emph{states} $\mathcal{S}$ (a state can be seen as a collection of variables) and a state-transition function $f: \mathcal{S} \times 2^{\mathcal{S}} \rightarrow \mathcal{S}$, where the first argument represents the state of a given node, and the second argument represents the (unordered) set of states observed at neighboring nodes.
Given its current state and the states of its neighbors, the state-transition function determines the new state a node adopts.
In this paper, we assume that the state-transition function  $f$ is \emph{random}, meaning that it takes as input an additional source of randomness, inducing  a probability distribution over the next state. Moreover, we assume that the nodes have the common knowledge of a parameter $N$, such that $N = \Omega(\log n)$, and that the set of states $\mathcal{S}$ and the state-transition function $f$ depend on $N$. 
The \emph{configuration} of a system is the product of the states of all nodes, and we denote with $\Gamma$ the set of all possible configurations of the system.

In this paper, we assume that each transition between configurations is driven by a \emph{synchronized scheduler}. This means that, in each transition step (called \emph{round}), all the nodes execute the state-transition function simultaneously and atomically.  That means that the evaluation of the state-transition function $f$ and the update of the states are done in one atomic step.

\paragraph{Self-stabilization} 
Let $(\mathcal S,f)$ be a randomized protocol executed under the synchronous scheduler, and let $P$ be a predicate on configurations.  
We say that $(\mathcal S,f)$ \emph{self-stabilizes to $P$} if the following two properties hold:

\begin{enumerate}
    \item \textit{Closure.}  
    The predicate $P$ is closed under the transition relation induced by $f$:  
    for every configuration $\gamma$ satisfying $P$, and every configuration $\gamma'$ such that \(\Pr[\gamma \rightarrow \gamma'] > 0\),
    it holds that $\gamma'$ also satisfies $P$.

    \item \textit{Convergence.}  
    For every initial configuration $\gamma_0$, with probability \(1\) there exists $\tau$ such that the configuration \(\gamma_\tau\) obtained after $\tau$ synchronous rounds satisfies $P$.
\end{enumerate}

Configurations \(\gamma\) satisfying $P$ are called legitimate configurations.

\paragraph{Silent self-stabilizing algorithms.}
A self-stabilizing algorithm is said to be \emph{silent} if each of its executions reaches a point in time after which the states of nodes do not change. 
A foundational result regarding space complexity in the context of self-stabilizing silent algorithms has been proved by Dolev et al. \cite{DolevGS99}, which shows that in
$n$-node networks, $\Omega(\log n)$ bits of memory per node are required for silently solving leader election.

\paragraph{State space and Memory Requirements}
One performance measure of a self-stabilizing algorithm is its \emph{state space} $|\mathcal{S}|$.  
An algorithm with $|\mathcal{S}|$ states requires $\Theta(\log |\mathcal{S}|)$ bits of memory per node.

\paragraph{Randomized convergence time.}
As discussed above, randomization is modeled by allowing the state-transition
function $f$ to depend on random bits. Formally, $f$ can be viewed as a mapping
from $\mathcal{S} \times 2^{\mathcal{S}}$ together with an additional source of
randomness, which induces a probability distribution over successor states.
In our algorithm, randomness is supplied locally at each execution round: every
node has access to a constant number of fresh random bits at each step.
Consequently, the convergence time $\tau$ for a fixed initial configuration \(\gamma_0\)
is not deterministic but a random variable defined over executions.

We consider the convergence of \emph{Las Vegas} type: starting from any initial
configuration, the system reaches a configuration satisfying $P$ with
probability~1, while the convergence time $\tau$ is a random variable. By
closure of $P$, once such a configuration is reached, all subsequent
configurations also satisfy $P$ with probability~1.  
We say that convergence occurs within time $T$ \emph{with high probability (w.h.p.)} 
if the random convergence time $\tau$ satisfies \(Pr[\tau \leq T] \geq 1 - n^{-c}\) for some constant $c = \Omega(1)$.

\paragraph{Leader Election Problem}

We study the \emph{eventual leader election} problem, i.e., the leader election problem without the additional silent property. Once a leader is elected, nodes may continue updating their states and communicating, but a unique leader is always maintained. A formal definition follows.

\begin{definition}[Eventual Leader Election]
An algorithm $(\States,f)$ is said to solve \emph{eventual leader election} in time $T\in\bbN$ if there exists a distinguished subset of states $L \subseteq \States$ such that, for all execution, there exists a node $v^* \in V$ such that, for every round $t \geq T$, $v^*$ is the unique node whose state belongs to $L$ at time $t$.  
\end{definition}

Note also that, with this definition of leader election, nodes do not need to detect when convergence occurs.

\subsection{Our Result}

In this paper, we prove the following result.

\begin{theorem}
\label{thm:main}
For any $n$-node connected graph $G=(V,E)$, if the nodes share the common knowledge of a parameter $N =\Theta(\log n)$, there exists a self-stabilizing randomized algorithm $\mathcal{A}$ such that:
\begin{itemize}
    \item Each node uses $\cO(\log n)$ states, i.e. $\cO(\log \log n)$ bits of memory; 
    \item Each node uses two random bits per round;
    \item The algorithm $\mathcal{A}$ stabilizes in $\cO(\mathrm{poly}(n))$ rounds with high probability.
\end{itemize}
\end{theorem}

More precisely, we prove that if the nodes have a common knowledge of a parameter \linebreak\mbox{$N \ge \max\{5,1+\log n\}$}, the nodes need $\Theta(\log N)$ bits of memory and $\Theta(N)$ states, and the algorithm stabilizes in $\cO(2^{3N}\log n)$ rounds w.h.p. In particular, when $N = \log n + \cO(1)$, we have convergence in time $\cO(n^3 \log n)$.

\subsection{Our Algorithm: an Informal Description}
To solve the self-stabilizing leader election problem, an algorithm $\mathcal{A}$ must simultaneously handle two key properties.  
The first is the \emph{verification} property, which is required for self-stabilization: leaders must propagate their presence throughout the network to prevent other nodes from becoming leaders, and nodes must be able to detect when no leader is present.  
The second is \emph{symmetry breaking}, which becomes necessary when multiple leaders are present simultaneously; nodes must break these symmetries to eventually elect a unique leader.

\paragraph{Verification by trains.}
Our algorithm is based on the idea of \emph{informative trains}\footnote{
While trains of information have appeared in prior work~\cite{korman2015fast}, those constructions use $\Omega(\log n)$ bits per node and serve a different purpose. In contrast, our approach achieves verification using only $\Theta(\log \log n)$ bits per node.} to handle the verification phase.  
In each round, leaders generate a \emph{train} of information that propagates through the network. Specifically, each leader creates a train consisting of $N$ \emph{wagons}, where $N$ is the common parameter known to all nodes ($N = \Theta(\log n)$, with $N \ge \max\{5,1+\log n\}$).
Each wagon carries minimal information: its position in the train (encoded using $\log N$ bits) and a single \emph{bit}. The wagons propagate through the network along a BFS-like traversal. Consecutive wagons can be interpreted as a binary string of length $N$, which can collectively store up to $N$ bits of information.
When created, these strings contain all zeros.

Each node actually carries two wagons. At each round, wagons shift by one position: the first wagon goes to the next node, while the last wagon moves to first position. This way, it takes two rounds for the back wagon to advance to the next node; this allows us to clean incomplete or erroneous trains from the network: intuitively, the cleaning is twice as fast as the train, and can catch up to the head of the train.

A key property of our trains is that their value can be distributively updated as they traverse the network. Formally, each train $T$ is associated with a value $\Value(T)$, which is the binary value of its string. At each step, nodes increment this value by at least one, effectively performing a distributed binary counter. Since $N \geq 1+\log n$, the train can count modulo at least $n$, ensuring that it can traverse the network for more than $2D$ rounds -- which is necessary to traverse the whole network, since the train only advances by half a node at each round.
Furthermore, our trains are designed to \emph{disappear} after reaching their maximum value (after at most $2^N = \cO(\mathrm{poly}(n))$ rounds). This mechanism ensures that in the absence of leaders, trains detect an \emph{overflow}, vanish, and trigger the creation of new leaders as needed.

\paragraph{Legitimate configurations.}
Our algorithm converges towards legitimate configurations which can be described in the following way. Naturally, legitimate configurations contain a unique leader \(v^*\). Moreover, the trains form a BFS-like structure around \(v^*\): the state of a node depends only on its distance to \(v^*\) and, when considering sequences of wagons along the paths of this BFS structure, we obtain coherent trains. Moreover, the value of these trains is closely related to their distance to \(v^*\): intuitively, a train \(T\) that has been traveling away from \(v^*\) for \(k\geq 0\) rounds has been incremented \(k\) times, and satisfies \(\Value(T) = k\). This guarantees that the counter never reaches \(2^N\), which is necessary for the closure property of legitimate configurations: otherwise, some nodes would detect the overflow and become new leaders.

\paragraph{Symmetry breaking.}
For the symmetry-breaking component, our algorithm draws inspiration from the approach in \cite{VacusZ25} for the beeping model, while integrating it with our information trains.  
In \cite{VacusZ25}, leaders emit \emph{beep waves} that expand through the network; leaders crossed by waves are eliminated, and waves from different leaders that meet cancel each other. By assigning leaders different beep frequencies using randomization, this process gradually eliminates all but one leader.  
We adapt this idea to our trains by allowing leaders to generate \emph{marked} or \emph{unmarked} trains at random. Marked trains have priority over unmarked trains and can overtake them, while two marked trains that collide cancel each other. A leader receiving a marked train from another leader is eliminated (i.e. it stops being a leader), but a leader cannot eliminate itself using its own train. Each train carries an additional \emph{flag} indicating whether it is marked or unmarked, enabling the network to enforce these interactions.

\paragraph{The use of randomness and the convergence time}
A key feature of our algorithm is that each node uses only a constant number of random bits per round. Despite this limitation, nodes can generate rare events happening with sufficiently small probability: specifically, every $N$ rounds, a node can produce a Bernoulli random variable with parameter $2^{-\Theta(N)}$.
This is achieved by maintaining a clock modulo \(N\) at each node, along with a random variable. The random variable is set to \(1\) when the clock reaches \(0\), and reset to zero in each round with a constant probability: at the end of each $N$-round phase, the random variable remains nonzero with probability $2^{-\Theta(N)}$ (which equals $n^{-\Theta(1)}$ when $N = \Theta(\log n)$), effectively simulating a rare-event random variable using minimal memory.

We use this mechanism for marking trains: a leader generates a \emph{marked} train if its random variable is $1$. This ensures that, with inverse-polynomial probability, exactly one leader \(v^*\) emits a marked train while all others emit only unmarked trains for a sufficiently long time. This property is crucial for symmetry breaking, as it guarantees that one leader survives while all others are eliminated: \(v^*\) will be able to build around it the BFS structure of legitimate configurations, without any interference with marked trains emitted by other leaders. Moreover, because the event occurs with inverse-polynomial probability when $N = \Theta(\log n)$, the algorithm converges in polynomial time.

\subsection{Roadmap}

The paper is organized as follows. In \cref{sec:related-works}, we provide an overview of related work; in \cref{sec:algo}, we present a detailed description of our algorithm; and in \cref{sec:proofs}, we give the proofs of correctness and convergence of the algorithm, establishing \cref{thm:main}. In \cref{sec:conclusion}, we conclude by summarizing the problems left open by our work.

\section{Related Work}
\label{sec:related-works}
In what follows, we summarize previous work on self-stabilizing leader election (SSLE) algorithms in both anonymous networks and networks with unique IDs, with a focus on approaches that aim to optimize space complexity. We note the direct correspondence between memory and states: an algorithm using $m$ bits of memory requires $2^{\Theta(m)}$ states.

\paragraph{Anonymous Networks}
Self-stabilizing leader election (SSLE) has been extensively studied in anonymous settings across several distributed models. In such contexts, randomization is essential, as deterministic symmetry breaking is impossible~\cite{angluin1980local}.

In the state model, \cite{BeauquierGJ07} studies SSLE on anonymous rings of size $n$, proposing a randomized protocol that uses $\cO(1)$ bits of memory on average, but may require $\Omega(\log \log n)$ bits in the worst case ring size. 
In detail, the authors' solution to SSLE uses $\Theta(\log m_n)$ bits of memory, where $m_n$ denotes the smallest integer that does not divide $n$. This value is constant on average but, for certain values of $n$ (for example, $n = k!$ for some $k \geq 1$), it implies a memory requirement of $\Theta( \log m_n)= \Theta(\log \log n)$.
The protocol is not silent and operates under an arbitrary scheduler. 
\cite{higham1999self} presents a randomized self-stabilizing leader election algorithm for anonymous, uniform, synchronous, and unidirectional rings of arbitrary but known size. In their protocol, both the processor states and the messages have size $\cO(\log n)$, and the expected stabilization time is $\cO(n \log n)$. 
In \cite{Itkis1995fast}, the authors present a randomized self-stabilizing leader election protocol for anonymous networks of arbitrary topology under an asynchronous scheduler\footnote{The paper does not specify which fairness assumptions are made on the scheduler.} using a constant number of bits of memory \emph{per edge}.
In \cite{Ostrovsky94}, the authors present a SSLE algorithm for anonymous networks using $O(\log^* n)$ bits of memory \emph{per edge}, assuming an asynchronous weakly fair scheduler.

Regarding lower bounds, \cite{BeauquierGJ99} proves a non-constant lower bound on the memory space for SSLE on unidirectional rings, which is $\Omega(\log \log n)$ for some values of $n$,\footnote{For detail, see the above discussion on the upper bound of \cite{BeauquierGJ07}.} under an unfair scheduler, which holds for anonymous networks and even for randomized protocols. 
In~\cite{DolevGS99}, the authors show that any silent self-stabilizing leader election algorithm on an anonymous ring -- even if randomized and executed under a synchronous scheduler -- requires at least $\sqrt{n}$ states and $\Omega(\log n)$ bits of memory.

The Stone Age model is a very weak distributed computing model in which nodes are anonymous, have constant memory, and can only perform extremely limited local computations, typically interacting with neighbors through simple signals or finite-state automata. In this model, the self-stabilizing leader election has been considered in \cite{EmekK21}, where a randomized self-stabilizing solution is proposed, with state space $\cO(D)$ and stabilizing in $\cO(D\log n)$ interactions w.h.p. The algorithm assumes the knowledge of $D$, and has no termination detection. The protocol can be implemented in the Stone-Age model only when $D = \Theta(1)$, but can be implemented in the state-model for any value of $D$.

In the population protocol model, we have a system of $n$ anonymous agents where, at each round, a pair of agents chosen u.a.r. interacts. 
Several SSLE protocols have been proposed for the population protocol model. In
\cite{YokotaSM21}, a non-terminating protocol for directed rings is proposed, converging in $\Theta(n^2)$ expected steps and using $\cO(n)$ states, requiring the approximate common knowledge of $n$.  Angluin et al.~\cite{AngluinAFJ05} propose a family of SSLE constant-state protocols, each corresponding to a class of rings.  In
\cite{Chen2019self}, a SSLE protocol  is proposed for directed rings of all sizes and tori of arbitrary sizes. The protocol uses a constant number of states, and terminates within an exponential number of interactions in the worst case.  In
\cite{chen2020self}, a $\cO(\mathrm{poly}(k))$-states protocol is proposed for $k$-regular graphs.
The recent papers \cite{AustinBFGH25} and \cite{gasieniec2025improving} consider the SSLE on complete graphs and explicitly optimize the time complexity while carefully balancing the trade-off between space and time. However, both papers solve the SSLE problem via the ranking problem, so both require $\Omega(n)$ states.  
As regarding lower bounds, 
\cite{CaiIW12} proves a lower bound of $\Omega(n)$ states for SSLE in worst-case graphs. In \cite{AngluinAFJ05}, it is proved that there does not exist a self-stabilizing protocol with constant memory for leader election in interaction graphs with general topology.

The following papers focus on space-efficient solutions for leader election in general graphs, while they do not address self-stabilization. The paper \cite{beauquier2013self} proposes a population protocol that solves leader election using a constant number of states per node for general graphs. \cite{AlistarhRV25} studies leader election in general graphs, establishing a worst-case expected time complexity of \(\Theta(n^3)\) rounds and presenting a near time-optimal protocol that requires only \(\cO(\log^2 n)\) states per node. Finally, \cite{VacusZ25} introduces a constant-state leader election protocol in the beeping model, converging in time $\cO(D^2 \log n)$, where $D$ is the network diameter.


\paragraph{Networks with IDs.}
We summarize in this section the most recent work on self-stabilizing leader election in networks with unique identifiers, using the state model. 
This is not meant to be an exhaustive survey, as we focus on settings different from those considered in this paper.
Regarding lower bounds, it is known that in the deterministic setting, electing a leader in a self-stabilizing manner requires at least 
$\Omega(\log \log n)$ bits of memory per node~\cite{BlinFB23}. 
For general networks with unique identifiers, \cite{BlinT20} proposed an algorithm using 
$O(\log \Delta + \log \log n)$ bits of memory per node, where $\Delta$ is the maximum degree of the network. 
Several deterministic SSLE algorithms have also been proposed for specific topologies. In particular, \cite{BeauquierGJ99} presents a deterministic, silent SSLE protocol using constant memory on unidirectional, ID-based rings with bounded identifiers, while \cite{ItkisLS95} presents a deterministic constant-space SSLE protocol on uniform bidirectional asynchronous rings of prime size.

\section{Formal Description of Our Algorithm}
\label{sec:algo}
In this section, we formally describe our algorithm. In \Cref{sec:node-variables}, we define the state space $\States$, in \Cref{sec:trains} we describe the circulation of trains, in \Cref{sec:errors} we describe how leaders are created, in \Cref{sec:leaders} we describe the behavior of leaders, and finally in \Cref{sec:transition-function} we describe our algorithm.
We assume that all nodes have common knowledge of a value $N$ such that \(N \geq 1+\log n\) and \(N = \Theta(\log n)\). We also assume \(N > 4\).

For any graph $G=(V,E)$ and a node $v \in V$, we denote by $N(v)$ the neighborhood of $v$. 
For any integer $k\geq 1$ we write $a \% k$ to denote the remainder when $a$ is divided by $k$, i.e. $a \pmod{k}$. For any $h \geq 1$, we denote with $[h] = \{0,\dots, h-1\}$.

\subsection{Node Variables}
\label{sec:node-variables}

The state space $\States$ is defined as a finite collection of state variables indexed by a set  $\mathcal{K}$, where $|\mathcal{K}| = \Theta(1)$, which together constitute its state. Formally, the state space is the product space $
\States = \prod_{\ell \in \mathcal{K}} D_\ell,$
where $D_\ell$ is the domain of the variable $\ell \in \mathcal{K}$. 
Hence, a node’s state is specified by the values of its state variables. In particular, for any node $v \in V$ and any variable $\ell \in \mathcal{K}$ we denote by  $v.\ell$ the value of the variable $\ell$ in the local state of node $v$, and with $\State v$ the local state of the node $v$. Moreover, for any $\ell \in \mathcal{K}$, we will denote by $v.\ell_t$ the value of the variable $v.\ell$ at the end of the execution round $t$ of the algorithm. 
We denote by $\gamma_t$ the configuration of the system at time $t$, which is the collection of the states of all nodes in the graph $G=(V,E)$ at that round. Formally, a configuration $\gamma_t$ is a vector of all node states:
\[\gamma_t = (v.\ell_t, \ell \in \mathcal{K})_{v \in V}.\]
The set of all configurations $\mathcal{S}^{|V|}$ is denoted with \(\Gamma\). We next define the variables associated to each node.

\begin{definition}[Node Variables]
The \emph{local variables} associated with a node $v$ are the following:
\begin{description}[topsep=3pt, leftmargin=0pt, labelsep=4pt, itemsep=-1pt, font=\normalfont]
    \item[\emph{Random Variable:}] $v.\rand \in \{0,1\}$, which determines whether $v$ initializes a marked or unmarked train.
    \item[\emph{Leader Flag:}] $v.\leader \in \{0,1\}$, indicating whether $v$ is a candidate leader.
    \item[\emph{Stations:}] \(v\) contains two stations, $v.F$ and $v.L$, which contain first and last wagons carried by $v$, respectively. See \cref{def:wagon} for the definition of a wagon.
\end{description}
\end{definition}

In particular, each node is associated with two wagon variables, which, as specified later, belong to a train distributed across the nodes. Trains are of two types: marked and unmarked. They circulate over the network and, during their traversal, implement a binary addition. Each wagon consists of a tuple of variables, as defined below.

\begin{definition}[Wagon Variable]
A wagon variable $B$ is either empty, in which case  $B = \bot$, or a tuple with the following components:
\begin{description}[topsep=3pt, leftmargin=0pt, labelsep=4pt, itemsep=-1pt, font=\normalfont]
    \item[\emph{Index of the Wagon:}] $B.\id \in [N]$, representing the position of the wagon within the train.
    \item[\emph{Bit Value:}] $B.\bit \in \{0,1\}$, representing the bit carried by the wagon.
    \item[\emph{Flag:}] $B.\flag \in \{0,1\}$, indicating whether the wagon belongs to a marked train.
    \item[\emph{Carry:}] $B.\carry \in \{0,1\}$, representing the carry generated during binary addition.
\end{description}
\label{def:wagon}
\end{definition}

Note that each variable has at most \(N\) different values, which ensures a space complexity of \(\cO(\log N) = \cO(\log \log n)\) bits.

\subsection{Trains and Wagon Circulation}
\label{sec:trains}

\subsubsection{Definitions}

We define a train as a sequence of consecutive nonempty wagons carried by nodes, which consistent flags.

\begin{definition}[Train]
    For $k_1,k_2 \in [N]$ with $k_1 \leq k_2$ a train $T$ is a set of wagons $B_{k_1}, \dots, B_{k_2}$ such that 
    \begin{enumerate}[(a)]
        \item $B_i.\id =i$ for each $k_1 \leq i \leq k_2$;
        \item $B_i \neq \bot$ for every $i \in [N]$
        \item $B_i.\flag = B_j.\flag$ for each $k_1 \leq i, j \leq k_2$.
    \end{enumerate}
    \label{def:train}
\end{definition}

In our algorithm, a train is represented as a sequence of consecutive wagons carried by nodes and circulates through the network by moving from node to node. 
Trains are continuously created by leaders: each time a leader generates a new train, it randomly decides whether the train is marked or unmarked, based on the variable $v.\rand$. 
The trains then propagate through the network via all nodes, serving to disseminate the existence of leaders. 
The circulation rules are designed so that, in the absence of leaders, any existing train of length $N$ disappears after $\cO(2^N)$ rounds. 
This guarantees that when no leaders are present, the trains eventually vanish, triggering the creation of new leaders.

The use of trains provides a way to distribute the information carried by the $N$ wagons, while requiring only $\cO(\log N)$ states per node: indeed, each node stores, for the wagons it carries, its identification via the wagon number. A train $T$ is a set of $N$ (or less) consecutive wagons: the $N$ bits carried by them, hence, form an $N$-bit string, which can encode up to $2^N$ distinct $N$-bit strings. 
Thus, by distributing $N$ wagons among $N$ consecutive nodes on a path and by implementing a distributed procedure to update their values, we can realize a distributed counter capable of representing values up to $2^N$, using only $\cO(\log N)$ states per node. 
To formalize what it means for a train to store a number up to $2^N$, we associate a numerical \emph{value} with each train. 
The value of a train is the number corresponding to the lexicographic rank of the binary string formed by its consecutive wagons, among all strings of the same length, as defined formally below.

\begin{definition}[Value of a Train]
Let $k_1, k_2 \in [N]$ with $k_1\leq k_2$, and let $T = B_{k_1}\cdots B_{k_2}$ be a train. The \emph{value} of the train $T$ is given by
\[\Value(T) = 2^{-k_1}\sum_{j=k_1}^{k_2}2^{j}( B_j.\bit + 2 \cdot B_j.\carry).\]
\label{def:train-value}
\end{definition}

\begin{observation}
\label{obs:train-value}
Notice that, for any train  decomposition $T =(B_{k_1})(B_{k_1+1}\cdots  B_{k_2})$, we have that
\[\Value(T) = \Value(B_{k_1}) + 2\Value(B_{k_1+1} \cdots B_{k_2}) = B_{k_1}.\bit + 2B_{k_1}.\carry + 2\Value(B_{k_1+1}\cdots B_{k_2}).\]
This writing evidences the following useful facts.
\begin{align*}
    \Value(T) \% 2 &= B_{k_1}.\bit \\
    \left\lfloor \frac{\Value(T)}2\right\rfloor &= B_{k_1}.\carry + \Value(B_{k_1+1}\cdots B_{k_2})
\end{align*}
Moreover, our algorithm maintains the invariant that if a wagon \(B\) is the last wagon of a train (ie. \(B.\id = N-1\)), then we always have \(B.\carry = 0\). This means that:
\[
\Value(B) = B.\bit
\]
\end{observation}

\subsubsection{Train Circulation}

In this section, we specify how train circulation is implemented. In particular, we describe how a node determines which wagon to carry in the next step and how it updates the respective value. The rules of the update are designed so that each time a train $T$ advances by one step in the network, its value $\Value(T)$ is incremented by at least one. This property is crucial for our algorithm, as it ensures that, in the absence of leaders, a train cannot circulate in the network for more than $\cO(n)$ rounds.

\paragraph{The next wagon.} In this section, we describe how a node $v$ selects and forwards the next wagon of a train carried by one of its neighboring nodes.  
For a wagon $B$, we denote by $\Next(B)$ the index of the subsequent wagon:
\[
\Next(B) = (B.\id + 1) \%N,
\]
so that the index wraps around to $0$ if $B$ is the last wagon of the train ($B.\id = N-1$), and is $B.\id + 1$ otherwise.
A key property of our algorithm is that \emph{marked} trains can pass over \emph{unmarked} ones. Specifically, if the head of a marked train is present in the neighborhood of $v$, the node $v$ will allow it to pass over any unmarked train. Moreover, if $v$ already contains a marked wagon that is not the last in the train, it should expect the next wagon to be marked as well. To formalize this, we introduce the following predicate. For $v \in V$, let $\SuccIsMarked(v)$ be true if $v$ is expecting a marked train:
\begin{align*}
    \SuccIsMarked(v) = &\; (v.L.\flag = 1 \wedge v.L.\id \neq N-1) \\
    &\;\vee (\exists u \in N(v): u.F.\flag = 1 \wedge u.F.\id = 0).
\end{align*}
For every node $v$, we define $\subc^1(v)$ as the set of neighbors carrying the next marked wagon for $v$, i.e.
\begin{align*}
    \subc^1(v) = \{u \in N(v): &(u.F.\flag = 1) \wedge \\ &(v.L.\flag = 1 \wedge u.F.\id = \Next(v.L)) 
    \vee  (v.L.\flag = 0 \wedge u.F.\id = 0))\}
\end{align*}
In detail, $\subc^1(v)$ contains the head of a marked train in the case in which $v.L.\flag = 0$, and the subsequent marked wagon in the case in which $v.L.\flag = 1$.
Similarly, we define $\subc^0(v)$ as the set of neighbors carrying the next unmarked wagon for $v$:
\begin{equation}
    \subc^0(v) = \{ u \in N(v): u.F.\flag = 0 \wedge u.F.\id = \Next(v.L) \}.
\end{equation}
In particular, a node $u \in \subc^0(v)$ is such that $u.F$ can either be the first wagon of a new train (when $v.L.\id = N-1$ and $u.F.\id = 0$) or the next wagon of the same train.

\paragraph{Performing the binary addition. }

We define the function $\Add$ to update the value of a train as it moves through the network. The function operates locally on a wagon $B$, using information from the next wagon in the train, denoted by $B'$.
Intuitively, $\Add$ ensures that the train behaves like a distributed binary counter: whenever the train advances, its value is incremented by one. The increment is applied in the first wagon of the train, and if it cannot be completed immediately, any overflow is stored in the wagon’s carry. This carry is propagated along the train and applied as soon as possible, ensuring that the train’s value is updated correctly and consistently as it circulates.

\begin{definition}[Addition function]
\label{def:merge-and-add}
The $\Add$ function update the value of the block $B$, with respect to a value of the next block $B'$.
\begin{center}
\begin{minipage}{\textwidth}
\begin{algorithm}[H]
\SetKwFunction{FAdd}{Add}
\SetKwProg{Fn}{}{:}{}
\Fn{\FAdd{$B,B'$}}{
\eIf{\(B'.\id = 0\)}{
$B.\bit \gets (B'.\bit + 1)\% 2$\;
$B.\carry \gets \mathds{1}(B'.\bit + 1 = 2)$\;
}{
$B.\bit \gets (B'.\bit + B.\carry)\%2$\;
$B.\carry \gets \mathds{1}(B'.\bit + B.\carry = 2)$
}
}\textbf{end}
\end{algorithm}
\end{minipage}
\end{center}
\end{definition}

\begin{observation}
    \label{obs:add-one-bit}
    \(\Add\) performs a one-bit addition, in the following sense.
    Denote \(B''\) the value of \(B\) after the execution of \(\Add(B,B')\). Then:
    \begin{itemize}
        \item if \(B'.\id = 0\), then \(B''.\bit+2\cdot B''.\carry = B'.\bit + 1\);
        \item otherwise, \(B''.\bit + 2 \cdot B''.\carry = B'.\bit + B.\carry\).
    \end{itemize}
\end{observation}

\paragraph{Updating the wagon.}
We will now describe the \WagonUpdate~function, which determines how the two wagons $v.F$ and $v.L$ carried by a node $v$ are updated.

\begin{definition}[Wagon Update function]
\label{def:wagon-update}
The \WagonUpdate~function has as input the wagon carried by a node $v$, i.e. $v.F$ and $v.L$, together with the wagon carried by its neighbors and update the values of $v.F$ and $v.L$.
\begin{center}
\begin{minipage}{\textwidth}
\begin{algorithm}[H]
\SetKwFunction{FAdd}{\WagonUpdate}
\SetKwProg{Fn}{}{:}{}
\Fn{\FAdd{$v.F, v.L, \{w.L, w \in N(v)\}$}}{
\eIf{$\SuccIsMarked(v)$}{
\eIf{\(v.L.\flag = 1 \vee v.L.\id = N-1\)}{\(\Add(v.F,v.L)\)\;\tcp{Updates $v.F$}}{\(v.F \gets \bot\)\;}
$u \gets \argmax_{w \in \subc^1(v)}w.F.\bit$\;}
{
\(\Add(v.F,v.L)\)\;\tcp{Updates $v.F$}
$u \gets \argmax_{w \in \subc^0(v)}w.F.\bit$\;}
$\Add(v.L, u.F)$\;
\tcp{Updates $v.L$}
}\textbf{end}
\end{algorithm}
\end{minipage}
\end{center}
\end{definition}

\subsection{Error Predicates and Leader Creation}
\label{sec:errors}

The aim of this section is to define the error predicate $\err(v)$ associated with a node $v$. 
Our algorithm is designed so that whenever a node $v$ detects an error, it declares itself a leader. 
The predicate $\err(v)$ is defined as the disjunction of two sub-predicates, which correspond to \emph{local} and \emph{global} errors.

\emph{Local} errors depend solely on the variables stored at node $v$. 
The algorithm guarantees that local errors are never generated during execution and, if present initially, are eliminated within the first round.
Specifically, a node $v$ detects a local error if one of the following conditions holds: (1) $v$ carries no last wagon; (2) both wagons carried by $v$ are non-empty, but $v.L$ is not a successor of $v.F$; (3) both wagons carried by $v$ are non-empty, belong to the same train, but their flags differ; (4-5) one of the wagons carried by \(v\) overflows, meaning it is the last wagon of a train but still has a carry.
The local errors are formally defined as follows.

\begin{definition}[Local Errors]
For a node $v$, we define the following predicates:
\begin{align*}
  \err_1(v) &:= v.L = \bot, \\
  \err_2(v) &:= v.L \neq \bot \;\wedge\; v.F \neq \bot \;\wedge\; v.L.\id \neq (v.F.\id + 1)\%N, \\
  \err_3(v) &:= v.L \neq \bot \;\wedge\; v.F \neq \bot \;\wedge\; v.L.\id \neq 0 \;\wedge\; v.L.\flag \neq v.F.\flag, \\
  \err_4(v) &:= v.F.\id = N-1 \wedge v.F.\carry = 1, \\
  \err_5(v) &:= v.L.\id = N-1 \wedge v.L.\carry = 1, \\
\end{align*}
\end{definition}

Because our algorithm guarantees that local errors are eliminated after the first round, we will implicitly assume that the above predicates are always false.

The \emph{global} error conditions at a node $v$ depend on the variables of $v$ as well as the variables of its neighbors $u \in N(v)$. 
When a node detects a global error, it declares itself a leader.  Unlike local errors, global errors can be created during execution, for instance when nodes detect (directly or indirectly) the absence of leaders.  Specifically, a node $v$ detects a global error if at least one of the following conditions holds: (1) $v$ has no proper successors among the wagons $u.F$ carried by its neighbors $u \in N(v)$ that are consistent with both the flag and the wagon number in $v.L$; 
    (2) The next wagon carried by some $u.F$ has $u.F.\bit = 1$ while simultaneously $v.L.\carry = 1$. 
    The latter condition captures an \emph{overflow} scenario, which occurs when the train has already been incremented to its maximum value and cannot be incremented further: the leader creation in such case is essential to ensure that trains do not circulate for more than $\cO(2^{N})$ rounds. 
    The global errors are formally defined as follows. 

\begin{definition}[Global Errors]
For a vertex $v$, we define the following predicates:
\begin{align*}
\ErrSuccessor(v) &:= \subc^{\SuccIsMarked(v)}(v) = \emptyset \\
\ErrOverflow^L(v) & := v.L.\id = N-2 \wedge v.L.\carry = 1 \wedge \max_{u \in \subc^{\SuccIsMarked(v)}(v)}u.F.\bit = 1 \\
    & \quad \wedge v.L.\flag = \SuccIsMarked(v) \\
\ErrOverflow^F(v) & := v.F.\id = N-2 \wedge v.F.\carry = 1 \wedge v.L.\bit = 1 \\
\ErrOverflow(v) & := \ErrOverflow^L(v) \vee \ErrOverflow^F(v)
\end{align*}
\end{definition}

We define the error predicate $\err(v)$ to be true if any underlying error occurs and the node is not already a leader (since errors are only used to trigger leadership), i.e.,
\[\err(v) := \neg v.\leader \wedge\left( \vee_{i=1}^5 \err_i(v) \vee \ErrSuccessor(v) \vee \ErrOverflow(v) \right).\]

As previously mentioned, a leader is created at node $v$ whenever $\err(v) = 1$. The following function, $\NewLeader$, specifies how the creation of a leader resets the local variables of $v$. Note that, in defining the following function, node $v$ has access to a Bernoulli$(\frac{1}{4})$ random variable, which can be generated using only two random bits.
\begin{definition}[New Leader]
\label{def:new-leader}
The \NewLeader~function resets some variables associated with $v$. Let $X$ be a Bernoulli random variable with expectation $\frac{1}{4}$.
\begin{center}
\begin{minipage}{\textwidth}
\begin{algorithm}[H]
\SetKwFunction{FAdd}{New-Leader}
\SetKwProg{Fn}{}{:}{}
\Fn{\FAdd{}}{
$v.\leader \gets 1$\;
$(v.F.\id,v.F.\bit, v.F.\carry, v.F.\flag) \gets (0,1,0,0)$\;
$(v.L.\id,v.L.\bit, v.L.\carry, v.L.\flag) \gets (1,0,0,0)$\;
$v.\rand \gets X$\;
}\textbf{end}
\end{algorithm}
\end{minipage}
\end{center}
\end{definition}

\subsection{The Leaders' Behaviors}
\label{sec:leaders}

In this section, we describe the behavior of leaders. As previously mentioned, leaders disseminate their presence by creating trains of length $N$, generating one wagon per step. Every $N$ rounds, when a leader creates the head of a new train, it decides whether the train will be \emph{marked} or not. 
Marked trains are generated randomly and are used for \emph{symmetry breaking}, to ensure the elimination of other leaders. 

The randomness used for marking the trains is determined by the local variable $v.\rand$, which is updated so that, every $N$ steps, it produces a new independent Bernoulli random variable with expectation $1/4^N$. This is a key feature of the algorithm: it allows the generation of a very small-probability event over time using only a constant number of random bits per step. 
Specifically, $v.\rand$ is initially a Bernoulli random variable with parameter $1/4$ and is updated at each step according to a geometric distribution. After a full phase of $N$ rounds, the resulting value is a Bernoulli$(1/4^N)$.

\begin{definition}[Wagon Creation and Random Variable update]\label{def:wagon-update-leader} The following function specifies the behavior of a leader $v$ and, in particular, updates the value of the wagons $v.F$ and $v.L$, and  the random variable $v.\rand$. Let $X$ be a Bernoulli random variable with expectation $\frac{1}{4}$.
\begin{center} 
\begin{minipage}{\textwidth}
\begin{algorithm}[H]
\SetKwFunction{FAdd}{Wagon-Creation}
\SetKwProg{Fn}{}{:}{}
\Fn{\FAdd{$v.F, v.L, v.\rand$}}{
$v.F \gets \Add(v.F,v.L)$\; \tcp{Update $v.F$}
\eIf{$v.L.\id = N-1$}{
$(v.L.\id, v.L.\bit, v.L.\carry, v.L.\flag) \gets (0, 0, 0, v.\rand)$\;
$v.\rand \gets X$\;
}{
$(v.L.\id, v.L.\bit, v.L.\carry, v.L.\flag) \gets (v.L.\id +1, 0, 0, v.L.\flag)$\;
$v.\rand \gets v.\rand \cdot X$
}
}\textbf{end}
\end{algorithm}
\end{minipage}
\end{center}    
\end{definition}


\subsection{Putting the Algorithm Together: the State-Transition Function}
\label{sec:transition-function}

We first define a predicate that is used to eliminate candidate leaders. 
Specifically, a node $v$ carrying an unmarked train is \emph{eliminated} if it receives the head of a marked train from a neighbor:
\[
\IsKilled(v) := (v.L.\flag = 0) \;\wedge\; (\exists u \in N(v), \; u.F.\flag = 1 \;\wedge\; u.F.\id = 0).
\]

We are now ready to define the state-transition function $f : \mathcal{S} \times 2^{\mathcal{S}} \to \mathcal{S}$
which is applied simultaneously at all nodes in the network to update their states. 
Specifically, $f$ updates the state $\State v$ according to the following steps. If node $v$ detects an error, i.e., if the predicate $\err(v)$ defined in \cref{sec:errors} holds, then $v$ becomes a leader. If a leader $v$ is carrying an unmarked train and receives the head of a marked train, i.e., if $\IsKilled(v) = 1$, then it is eliminated as a leader. If $v$ has not been eliminated in the previous steps, it creates wagons according to the $\WagonUpdateLeader$ function described in \cref{sec:leaders}. Finally, if $v$ is not a leader, it updates its wagons according to the $\WagonUpdate$ function described in \cref{sec:trains}.
The algorithm is formally defined below.
\begin{definition}[State-transition function]
At any round $t \geq 0$, for each node $v \in V$, the state $\State v$ is updated according to the following algorithm.
    \begin{center}
\begin{minipage}{\textwidth}
\begin{algorithm}[H]
\SetKwFunction{FAdd}{Update-State}
\SetKwProg{Fn}{}{:}{}
\Fn{\FAdd{$\State v$}}{
\eIf{$\err(v)$}{
$\NewLeader()$\;}{
\If{$\IsKilled(v)$}{$v.\leader \gets 0$}
\eIf{$v.\leader = 1$}{
{
$\WagonUpdateLeader(v.F, v.L, v.\rand)$\;
}
}{
$\WagonUpdate(v.F,v.L, \{w.L, w \in N(v)\})$\;
}}
}\textbf{end}
\end{algorithm}
\end{minipage}
\end{center}  
\end{definition}

\section{Proofs}
\label{sec:proofs}
We start by defining the set of legitimate configurations. Two properties are required: legitimate configuration must contain a unique leader, and legitimate configurations must be closed, in the sense that a transition from a legitimate configuration \(\gamma\) leads to a new legitimate configuration \(\gamma'\) with probability $1$. In addition, we require that the unique leader is the same in both \(\gamma\) and \(\gamma'\). Roughly, we define legitimate configurations as configurations in which trains form a BFS pattern around \(v^*\): all trains are emitted by \(v^*\) and travel away from it, increasing their counters as they travel. The value of a train is closely related to its distance from \(v^*\), guaranteeing that it will never reach \(2^N > D\), which would cause an overflow and thus the creation of a new leader. Let us now define this notion formally.

Let $v^* \in V$ be a distinguished node. Define the \emph{layers} with respect to $v^*$ as
\[L_{2i} = \{v.L: d(v,v^*) = i\}\]
\[L_{2i+1} = \{v.F: d(v,v^*) = i\}\]
Hence, intuitively, each layer $L_j$ at time $t$ is carrying the wagon that $v^*$ sent at time $t-j$.
\begin{definition}[Legitimate Configuration]\label{def:legitimate} A legitimate configuration is such that
\begin{enumerate}[(a)]
    \item \label{item:leg:leader} There exists a unique leader $v^*$;
    \item \label{item:leg:layers} Each layer $L_i$ with respect to \(v^*\) is a singleton, whose unique element \(B_i\) satisfies
    \[(B_i.\id + i)\% N = B_0.\id;\]
    \item \label{item:leg:train-value} For all \(k \ge 0\) \(T = B_k, \dots, B_{k-\min\{k,N-1-B_k.\id\}}\) is a partial train and:
    \[
        \Value(T) = \left\lfloor \frac{k}{2^{B_k.\id}} \right\rfloor
    \]
\end{enumerate}
Denote \(\legitimate\) the set of legitimate configurations.
\end{definition}

\begin{lemma}
\label{lem:leg-closed}
The set \(\legitimate\) of legitimate configuration is closed. Moreover, for any configurations \(\gamma,\gamma'\) such that \(\Pr[\gamma\rightarrow \gamma']>0\), the unique leader of \(\gamma'\) is the same as that of \(\gamma\).
\end{lemma}

\begin{proof}
    Let \(\gamma \in \legitimate\) and \(\gamma' \in \gamma\) such that there is a transition \(\Pr[\gamma \rightarrow \gamma']>0\). We use the same notations \(v^*, L_i, B_i\) as in \Cref{def:legitimate}.

    \begin{claim}
        \(v^*\) is still a leader in \(\gamma'\).
    \end{claim}
    \begin{proof}\triangleqed
        It suffices to prove that \(\IsKilled(v^*)\) is false.
        Let \(u \in N(v^*)\), it holds that \(\gamma(u).F = B_3\). Assume \(B_3.\flag = 1\) and \(B_3.\id = 0\). \(\gamma \in \legitimate\), so \(B_3, B_2, \dots, B_{3-\min\{3,N-1\}}\) is a train. In particular, since \(N > 4\), it holds that \(v^*.L.\flag = B_0.\flag = B_3.\flag = 1\). Thus, \(\IsKilled(v^*)\) is false.
    \end{proof}

    \begin{claim}  \label{claim:leg-closed:no-other-leader}
In \(\gamma'\), there is no leader other than \(v^*\).
    \end{claim}
    \begin{proof}\triangleqed
        It suffices to prove that, for all \(u \ne v^* \in V\), \(\err(u)\) is false. Predicate \(\ErrSuccessor(u)\) states that \(u\) sees a correct ``next wagon'' among its neighbors. Let \(k = d(u,v^*)\) and \(w_0\in N(u)\) such that \(d(w_0,v^*) = k-1\). Remark that $u.L = B_{2k}$ and $u.F = B_{2k+1}$. By definition of legitimate configurations, we have \(w_0.F.\id = (u.L.\id +1)\%N\), and if \(u.L.\id < N-1\) then we also have \(w_0.F.\flag = u.L.\flag\). This means that \(w_0.F\) is always a correct next wagon for \(u\), and thus \(\ErrSuccessor(u)\) is false.

        In fact, we can prove something even stronger: \(w_0.F = B_{2k-1}\) is the \emph{only} correct next wagon for \(u\).
        Indeed, by definition of the layers, for all neighbor \(w\in N(u)\) we have \(w.F \in \{B_{2k-1},B_{2k+1},B_{2k+3}\}\).
        Remark that, since \(N > 4\), for \(j=\{1,3\}\) we have that \(B_{2k+j}.\id = (B_{2k}.\id - j)\%N \ne (B_{2k}.\id+1)\%N\).
        Thus, the only way for \(B_{2k+j}\) to be a correct next wagon for \(u\) is if \(B_{2k+j}.\id = 0\) and \(B_{2k+j}.\flag = 1\).
        But then, because \(N > 4\), \(B_{2k+j}, \dots, B_{2k}\) is a partial train (by \Cref{def:legitimate}\ref{item:leg:train-value}), so we also have \(B_{2k}.\flag = 1\) and \(B_{2k}.\id < N-1\). Thus \(B_{2k+j}\) is not a correct next wagon for \(u\).
        
        With this, we can prove that \(\ErrOverflow(u)\) is false. It suffices to prove that, if \(B_{2k}.\id = N-2\), then \(B_{2k}.\carry + B_{2k-1}.\bit < 2\). From \Cref{def:legitimate}\ref{item:leg:train-value}, since $N \geq 1+\log n$, we have:
        \begin{align*}
            \Value(B_{2k}, B_{2k-1}) &= \left\lfloor \frac{k}{2^{N-2}} \right\rfloor \\
            &\leq \left\lfloor \frac{2n-1}{n/2}\right\rfloor\\
            &< 4
        \end{align*}
        But, by definition, we also have \(\Value(B_{2k},B_{2k-1}) \geq 2\cdot B_{2k}.\carry + 2\cdot B_{2k-1}.\bit\). Thus, \(B_{2k}.\carry + B_{2k-1}.\bit < 2\), and \(\ErrOverflow(u)\) is false.
    \end{proof}

    The fact that \(\gamma'\) satisfies \Cref{def:legitimate}\ref{item:leg:layers} is a direct consequence from the remark that the only correct next wagon for a node \(u\) at distance \(k\) from \(v\) is \(B_{2k-1}\). Denote \(B'_j\) the single element of the layers \(L'_i\) with respect to \(v^*\) in \(\gamma'\).

    It remains to show that \(\gamma'\) satisfies \Cref{def:legitimate}\ref{item:leg:train-value}. We prove it by induction on \(k\), it is immediate for \(k=0\). Let \(k > 0\), and assume \Cref{def:legitimate}\ref{item:leg:train-value} is satisfied in \(\gamma'\) for all \(j < k\). As mentioned in the proof of \Cref{claim:leg-closed:no-other-leader}, for all \(j>0\), if \(j=2i\) the only correct ``next wagon'' for nodes containing \(B_{2i}\) is \(B_{2i-1}\). If \(j = 2i+1\), then nodes containing \(B_{2i+1}\) also contain \(B_{2i}\), which is the only correct ``next wagon''. In both case, this means that \(B'_j = \Add(B_j, B_{j-1})\) (see \Cref{def:merge-and-add}). Then, there are three cases:
    \begin{itemize}
        \item \(B_{j-1}.\id \not\in\{0,N-1\}\). Then:
        \begin{align*}
            \Value(B'_j, \dots) &= (B'_j.\bit + 2 \cdot B'_j.\carry) + 2 \cdot \Value(B'_{j-1}, \dots)
            & \text{(\Cref{obs:train-value})}\\
            &= (B_j.\carry + B_{j-1}.\bit)
            + 2 \cdot \left\lfloor \frac{j-1}{2^{B'_{j-1}.\id}} \right\rfloor
            & \text{(\Cref{obs:add-one-bit})}\\
            &= B_j.\carry + \Value(B_{j-1},\dots)\%2
            + 2 \cdot \left\lfloor \frac{j-1}{2^{1+B_{j-1}.\id}} \right\rfloor
            & \text{(\Cref{obs:train-value})}\\
            &= B_j.\carry + \Value(B_{j-1},\dots)\% 2
            + 2 \cdot \left\lfloor \frac{\Value(B_{j-1},\dots)}{2} \right\rfloor \\
            &= B_j.\carry + \Value(B_{j-1}, \dots) \\
            &= \left\lfloor \frac{\Value(B_j, \dots)}{2} \right\rfloor
            & \text{(\Cref{obs:train-value})}\\
            &= \left\lfloor \frac{j}{2^{1+B_j.\id}} \right\rfloor \\
            &= \left\lfloor \frac{j}{2^{B'_j.\id}} \right\rfloor
        \end{align*}
        \item \(B_{j-1}.\id = 0\). Then, the situation is the same but we always add \(1\), not taking into account \(B_j.\carry\):
        \begin{align*}
            \Value(B'_j, \dots) &= (B'_j.\bit + 2 \cdot B'_j.\carry) + 2 \cdot \Value(B'_{j-1}, \dots) \\
            &= (B_{j-1}.\bit+1) + 2 \cdot \left\lfloor \frac{j-1}2\right\rfloor
            & \text{(\Cref{obs:add-one-bit})}\\
            &= (j-1)\% 2 + 1 + 2 \cdot \left\lfloor \frac{j-1}2\right\rfloor
            & \text{(\Cref{obs:train-value})}\\
            &= j
        \end{align*}
        \item \(B_{j-1}.\id = N-1\). Then, the expression of \(\Value(B'_j)\) is simplified:
        \begin{align*}
            \Value(B'_j, \dots) &= B'_j.\bit \\
            &= B_j.\carry + B_{j-1}.\bit \\
            &= \left\lfloor \frac{\Value(B_j, \dots)}2 \right\rfloor \\
            &= \left\lfloor \frac{j}{2^{1+(N-2)}} \right\rfloor
        \end{align*}
    \end{itemize}
\end{proof}

Now that we proved that legitimate configurations are closed, it remains to show that we eventually reach a legitimate configuration. The next lemma states that, if a leader \(v^*\) emits a marked train while there is no other marked wagon in the system, and if no other leaders emit marked trains for a sufficiently long time, we reach a legitimate configuration. Intuitively, the BFS structure of legitimate configurations will grow around \(v^*\) as the marked train it emitted advances and eliminates the other leaders. After \(k\) rounds, the BFS structure reaches the first \(k\) layers with respect to \(v^*\).

\begin{lemma}
    \label{lem:leg-grow}
    Consider an execution \(\execution\) and let \(t> 0\) be a time such that there is no marked wagon at time \({t-1}\). Assume that some leader \(v^*\) starts emitting a marked train at time \({t}\) (meaning \(v^*.L.\flag_{t} = 1, v^*.L.\id_t = 0\)), and that no other leader emits a marked wagon in the next \(k\) rounds, for some \(k \leq 2\epsilon(v^*)+1\) (where \(\epsilon(v^*)\) denotes the eccentricity of \(v^*\), ie. \(\epsilon(v^*) = \max_{u \in V} d(u,v^*)\)). Then, at time \({t+k}\):
    \begin{enumerate}[(i)]
        \item \label{item:leg-grow:leader} for all \(u \neq v\) such that \(d(u,v) \leq \frac k 2\), \(u\) is not a leader at time \({t+k}\);
        \item \label{item:leg-grow:layers} for all \(i \le k\), each layer \(L_i\) with respect to \(v^*\) is a singleton, whose unique element \(B_i\) satisfies
        \[
        (B_i.\id+i)\%N = B_0.\id;
        \]
        \item \label{item:leg-grow:train-value} for all \(i \le k\), \(T= B_i, \dots, B_{i-\min\{i, N-1-B_i.\id\}}\) is a partial train and:
        \[
        \Value(T) = \left\lfloor \frac{i}{2^{B_i.\id}}\right\rfloor;
        \]
        \item \label{item:leg-grow:no-marked} for all \(i > k\) there are no marked wagons in layer \(L_i\) with respect to \(v^*\);
        \item \label{item:leg-grow:marked} \(B_k.\id = 0\) and \(B_k.\flag = 1\).
    \end{enumerate}
    In particular, if \(k=2\epsilon(v^*)+1\), \(\gamma_{t+k} \in \legitimate\).
\end{lemma}

\begin{proof}
    We use the same notations as in the lemma. We prove the result by induction on \(k\). The initialization is immediate for \(k=0\). We assume the result is true for \(k\) and we prove it for \(k+1\).

    \begin{claim}\label{claim:leg-grow:next-wagon}
        Let \(i \in [1, k+1]\) and \(B' \in L_i\). The only correct next wagon for \(B'\) is \(B_{i-1}\).
    \end{claim}
    \begin{proof}\triangleqed
        If \(B' = v.F\) for some \(v \in V\) (ie. \(i\%2 = 1\)), the only correct next wagon is \(v.L \in L_{i-1}\) so \(v.L = B_{i-1}\). If \(B' = v.L\) for some \(v \in V\) (ie. \(i\%2 = 0\)), the candidate next wagons are \(\{u.F \mid u \in N(v)\}\). Observe that it is included in \(L_{i-1} \cup L_{i+1} \cup L_{i+3}\). If \(i+3 > k\) (resp. \(i+1 > k\)), then \(k-i < 3 < N-1\). In particular:
        \begin{itemize}
            \item by \ref{item:leg-grow:marked} and \ref{item:leg-grow:train-value}, \(B'\) is in the same train as \(B_k\) which is marked, so it is marked;
            \item \(B'\) is marked and \(B'.\id < N-1\), so correct next wagons are necessarily marked;
            \item by \ref{item:leg-grow:no-marked}, \(L_{i+3}\) (resp. \(L_{i+1}\)) does not contain any marked wagon.
        \end{itemize}
        This implies that correct next wagons for \(B'\) are included in layers \(L_j\), with \(j \le k\). But these layers satisfy \ref{item:leg-grow:leader}\ref{item:leg-grow:layers}\ref{item:leg-grow:train-value}, so \(B'\) essentially behaves as in a legitimate configuration. By the same analysis as in the proof \Cref{claim:leg-closed:no-other-leader}, the only correct next wagon for \(B'\) is \(B_{i-1}\).
    \end{proof}

    We now prove each item of the lemma.
    \begin{itemize}
        \item[\ref{item:leg-grow:leader}] If \(k+1\) is even then, for each node \(v\) such that \(d(v,v^*) = \frac{k+1}2\), \(v.L\) is unmarked (by \ref{item:leg-grow:no-marked}) and the only correct next wagon is \(B_k\) (by \Cref{claim:leg-grow:next-wagon}) which satisfies \(B_k.\flag = 1\) and \(B_k.\id = 0\) (by \ref{item:leg-grow:marked}). This implies that \(\IsKilled(v)\) is true and \(\err(v)\) is false, so \(v\) is not a leader at time \({t+k+1}\). For all other nodes, or if \(k+1\) is odd, the situation is the same as in a legitimate configuration, one may refer to the proof of \Cref{lem:leg-closed}.
        \item[\ref{item:leg-grow:layers}\ref{item:leg-grow:train-value}] The situation is the same as in a legitimate configuration, one may refer to the proof of \Cref{lem:leg-closed}.
        \item[\ref{item:leg-grow:no-marked}] Since no marked wagon is emitted by a leader other than \(v^*\), the only way for a station in \(L_i, i>k+1\) to contain a marked wagon at time \({t+k+1}\) is if it had a correct next marked wagon at time \({t+k}\). But, as mentioned in the proof of \Cref{claim:leg-grow:next-wagon}, the candidate next wagons are included in \(L_{i-1} \cup L_{i+1} \cup L_{i+3}\), which did not contain any marked wagon at time \(t+k\) by \ref{item:leg-grow:no-marked}.
        \item[\ref{item:leg-grow:marked}] This is a direct consequence of \Cref{claim:leg-grow:next-wagon} for \(i = k+1\), in conjunction with \ref{item:leg-grow:marked}.
    \end{itemize}
\end{proof}

In order to show that the situation described in \Cref{lem:leg-grow} happens, two things are required: the marked wagons must disappear, and, if there is no leader, one must be created. Both are consequences of the values of the trains reaching their maximum, \(2^N\). Before anything else, we need to formalize what are the trains present in a configuration \(\gamma\). Roughly, a train \(T\) is present in \(\gamma\) if there is a series of adjacent stations which contain the consecutive wagons of \(T\). In the formal definition, we need to differentiate the two kinds of stations.

\begin{definition}
    Let \(\gamma \in \Gamma\). The set \(\Trains(\gamma)\) is defined as the set of (complete) trains \(T_0\cdots T_{N-1}\) such that one of the following is true:
    \begin{itemize}
        \item there exists a sequence of nodes \(v_0, \dots, v_{\lfloor(N-1)/2\rfloor}\in V\) such that
        \begin{align*}
        v_k.F &= T_{2k} & \text{for all } k \text{ such that } 2k < N\\
        v_k.L &= T_{2k+1} & \text{for all } k \text{ such that } 2k+1 < N
        \end{align*}
        \item there exists a sequence of nodes \(v_0, \dots, v_{\lfloor(N-1)/2\rfloor}\in V\) such that
        \begin{align*}
        v_k.F &= T_{2k-1} & \text{for all } k \text{ such that } 2k-1 < N\\
        v_k.L &= T_{2k} & \text{for all } k \text{ such that } 2k < N
        \end{align*}
    \end{itemize}
    In both cases, \(v_k.F\) and \(v_k.L\) are called the stations containing \(T_{2k}, T_{2k+1}\) or \(T_{2k-1}\).

    We define \(\Trains^i(\gamma)\) as the subset of \(\Trains(\gamma)\) containing only marked trains (for \(i=1\)) or unmarked trains (for \(i=0\)).
\end{definition}

The next lemma shows that the value of each train is incremented at each round: this is the core idea of trains behaving as traveling counters. There are however two restrictions. First, we only show that the minimum value among trains is incremented by at least one. Second, we need to assume that no leader emits wagons with the same flag as the trains we are considering, otherwise a new train with very small value may appear. This is still sufficient to prove \Cref{lem:leader-creation,lem:marked-disparition} which state, respectively, that a leader is eventually created if there are no leaders, and that marked wagons eventually disappear if no new marked wagons are created by leaders.

\begin{lemma}
    \label{lem:train-incr}
    Consider an execution \(\execution\). Let \(t \in \bbN\) and \(i\in\{0,1\}\) and assume no leaders contain wagons \(B\) such that \(B.\flag = i\) at time \({t+1}\). Then:
    \[
        \min_{T'\in\Trains^i(\gamma_{t+1})} \Value(T') \geq \min_{T \in \Trains^i(\gamma_t)} \Value(T) + 1
    \]
\end{lemma}

\begin{proof}
    We fix an execution \(\execution\), at time \(t \in \bbN\) and a flag \(i \in \{0,1\}\). We assume no leaders contain wagons \(B\) such that \(B.\flag = i\) at time \({t+1}\).

    Let \(T' \in \Trains^i(\gamma_{t+1})\). Necessarily, at time \(t\), for each \(1 \le k \le N-1\), the station containing \(T'_k\) at time \(t+1\) stored, at time \(t\), a wagon \(T_{k-1}\) such that \(T_{k-1}.\id = T_k.\id -1 = k-1\). Moreover, since the station containing \(T'_{N-1}\) is not a leader (since, at time \({t+1}\), no leaders contain wagons \(B\) such that \(B.\flag = i\)). Thus, at time \(t\), it had a correct ``next wagon'' \(T_{N-1}\) in a adjacent station. By construction, we have \(T \in \Trains^i(\gamma_t)\). We now show that \(\Value(T') \ge \Value(T)+1\).

    \begin{claim}
        \label{claim:train-incr:advance-divide}
        Let \(k < N-1\). Then:
        \[
        \Value(T'_{k+1}, \dots, T'_{N-1}) \ge \left\lfloor\frac{\Value(T_k, \dots, T_{N-1})}2\right\rfloor
        \]
    \end{claim}
    \begin{proof}\triangleqed
        By (reversed) induction on \(k\). For \(k=N-2\):
        \begin{align*}
            \Value(T'_{N-1}) &= T'_{N-1}.\bit \\
            &\ge T_{N-2}.\carry + T_{N-1}.\bit &\text{(\Cref{obs:add-one-bit})}\\
            &= \left\lfloor\frac{\Value(T_{N-2}, T_{N-1})}2\right\rfloor&\text{(\Cref{obs:train-value})}
        \end{align*}
        Now let \(k < N-2\). Assume the claim is true for \(k+1\), we prove it for \(k\).
        \begin{align*}
            \Value(T'_{k+1}, \dots, T'_{N-1}) &= T'_{k+1}.\bit + 2\cdot T'_{k+1}.\carry + 2\cdot\Value(T'_{k+2}, \dots, T'_{N-1}) \\
            &\ge T_k.\carry + T_{k+1}.\bit + 2\cdot \left\lfloor \frac{\Value(T_{k+1}, \dots, T_{N-1})}2\right\rfloor & \text{(Induction)}\\
            &= T_k.\carry + \Value(T_{k+1}, \dots, T_{N-1}) \\
            &= \left\lfloor \frac{\Value(T_k, \dots T_{N-1})}2 \right\rfloor & \text{(\Cref{obs:train-value})}
        \end{align*}
    \end{proof}

    Finally, observe that:
    \begin{align*}
        \Value(T') &= T'_0.\bit + 2 \cdot T'_0.\carry + 2 \cdot \Value(T'_1, \dots, T'_{N-1}) \\
        &\ge 1 + T_0.\bit + 2\cdot \left\lfloor\frac{\Value(T_0, \dots, T_{N-1})}2\right\rfloor & \text{(\Cref{claim:train-incr:advance-divide})} \\
        &= 1 + \Value(T)
    \end{align*}
\end{proof}

The next lemma states that, in configurations without any leader, a leader is eventually created. The proof relies on the fact that, by \Cref{lem:train-incr}, the values of the trains eventually reach \(2^N\), at which point the overflow triggers \(\ErrOverflow\) which in turn creates a leader.

\begin{lemma}
    \label{lem:leader-creation}
    Consider an execution \(\execution\) and let \(t\in \bbN\). Then, for some integer \(k < 2^N + N\), there is at least one leader at time \({t+k}\).
\end{lemma}

\begin{proof}
    If, for some \(k < 2^N + N\), we have \(\Trains(\gamma_{t+k}) = \emptyset\), then necessarily there is a node \(v\in V\) that does not have any correct ``next wagon''. In particular, \(\ErrSuccessor(v)\) is true, and \(v\) will become a leader in the next round. We now assume \(\Trains(\gamma_{t+k}) \ne \emptyset\), for all \(k < 2^N+N\).
    
    Assume by contradiction that there is no leader in any \(\gamma_{t+k}, k \le 2^N+N\). By \Cref{lem:train-incr}, we have:
    \[
    \min_{T \in \Trains(\gamma_{t+2^N})} \Value(T) \ge 2^N
    \]
    Assume for now that \(\Trains^1(\gamma_{t+2^N})\ne \emptyset\) and let \(T = (T_0\cdots T_{N-1}) \in \Trains^1(\gamma_{t+2^N})\). Consider the station containing \(T_0\). After \(N-2\) rounds, it will contain a wagon \(T'_{N-2}\), and a neighboring station will contain a wagon \(T'_{N-1}\) such that \(T'_{N-2}, T'_{N-1}\) form a partial train. By iterating \Cref{claim:train-incr:advance-divide}, we have that \(\Value(T'_{N-2},T'_{N-1}) \ge \left\lfloor\frac{\Value(T)}{2^{N-2}}\right\rfloor\). Then:
    \begin{align*}
        T'_{N-2}.\carry + T'_{N-1}.\bit &= \left\lfloor \frac{\Value(T'_{N-2}, T'_{N-1})}2 \right\rfloor & \text{(\Cref{obs:train-value})}\\
        &\ge \left\lfloor \frac{\Value(T)}{2\cdot 2^{N-2}} \right \rfloor \\
        &\ge \left\lfloor \frac{2^N}{2^{N-1}} \right\rfloor \\
        &=2
    \end{align*}
    This implies that the node \(v\in V\) containing \(T'_{N-2}\) will satisfy \(\ErrOverflow(v)\) at time \(t+2^N+N-2\), and will become a leader.

    We now come back to the case \(\Trains^1(\gamma_{t+2^N})= \emptyset\). Since there are no leaders, no marked trains will be created. Also, there is no marked wagon \(B\) with \(B.\id = 0\) in the system, otherwise there would be at least one marked wagon without a correct successor, which would create a new leader in the next round thanks to \(\ErrSuccessor\). Thus, no marked train can cut an unmarked train. This implies that the above analysis for the case \(\Trains^1(\gamma_{t+2^N})\ne \emptyset\) works in the exact same way for unmarked trains.
\end{proof}

The next lemma states that, when leaders do not emit marked wagons, all marked wagons eventually disappear.

\begin{lemma}
    \label{lem:marked-disparition}
    Consider an execution \(\execution\) and let \(t\in \bbN\) be a time such that no marked wagon is emitted by a leader between \(\gamma_t\) and \(\gamma_{t+2^{N}+2N-2}\). Then there is no marked wagon at time \({t+2^{N}+2N-2}\).
\end{lemma}

As for the previous lemma, the proof relies on the fact that, by \Cref{lem:train-incr}, the values of the trains eventually reach \(2^N\), at which point the overflow triggers \(\ErrOverflow\). However, \(\ErrOverflow\) only eliminates the last wagon of the train: we need one additional step to show that the rest of the marked wagons are progressively eliminated by lack of a correct next wagon, ie. \(\ErrSuccessor\). These two steps correspond to the two claims of the proof below.

\begin{proof}[Proof of \cref{lem:marked-disparition}]
    Let \(\execution\) and \(t\in\bbN\) as in the statement of the lemma.

    \begin{claim}
        \label{claim:marked-disparition:overflow}
        There is no marked wagon \(B\) satisfying \(B.\id = N-1\) at time \({t+2^N+N-1}\).
    \end{claim}
    \begin{proof}\triangleqed
        Assume by contradiction that there is a marked wagon \(B\) satisfying \(B.\id=N-1\) at time \({t+2^N+N-1}\). Then, the station containing \(B\) at time \(t+2^N+N-1\) contained at time \(t+2^N\) a marked wagon \(T_0\) which was the first wagon of a train \(T\in\Trains^1(\gamma_{t+2^N})\). Since no marked wagon is emitted, by \Cref{lem:train-incr} we have \(\Value(T) \ge 2^N\). By the same analysis as in the proof of \Cref{lem:leader-creation}, at round \(t+2^N+N-2\) the station will satisfy \(\ErrOverflow\), so it will become a leader at time \(t+2^N+N-1\). This contradicts the initial hypothesis that this same station contains a wagon satisfying \(B.\id = N-1\) at time \({t+2^N+N-2}\).
    \end{proof}

    \begin{claim}
        \label{claim:marked-disparition:induction}
        For all \(0 \le k \le N-1\), there is no marked wagon \(B\) satisfying \(B.\id \ge N-1-k\) at time \({t+2^N+N-1+k}\).
    \end{claim}
    \begin{proof}\triangleqed
        By induction on \(k\). The case \(k=0\) is exactly \Cref{claim:marked-disparition:overflow}. Let \(k \geq 0\), we prove the claim for \(k+1\).
        Since no leaders emit marked wagons, the only way for a marked wagon \(B\) with \(B.\id \ge N-1-(k+1)\) to be present at time \({t+2^N+N-1+(k+1)}\) is if, at time \({t+2^N+N-1+k}\), there were:
        \begin{itemize}
            \item a marked wagon \(B_0\) in the same station as \(B\), satisfying \(B_0.\id = B.\id-1\);
            \item a correct next (marked) wagon \(B_1\) in a neighboring station, satisfying \(B_1.\id = B.\id\).
        \end{itemize}
        By induction hypothesis, \(B.\id = B_1.\id < N-1-k\), thus \(B.\id = N-1-(k+1)\). There are two cases.
        \begin{itemize}
            \item There is a node \(v\in V\) such that, at time \({t+2^N+N-1+k}\), \(v.F = B_0\) and \(v.L = B_1\). But then by induction hypothesis there is no correct next wagon for \(B_1\), so \(v\) satisfies \(\ErrSuccessor(v)\) and becomes a leader at time \({t+2^N+N-1+(k+1)}\). This contradicts the fact that it contains a marked wagon \(B\).
            \item There is a node \(v \in V\) and a neighbor \(u \in N(v)\) such that, at time \({t+2^N+N-1+k}\), \(v.L = B_0\) and \(u.F = B_1\). Since \(B_1.\id = B.\id < N-1\), we have \(u.L.\flag = 1\) and \(u.L.\id = B_1.\id + 1 \ge N-1-k\). This contradicts the induction hypothesis.
        \end{itemize}
    \end{proof}

    The lemma is exactly \Cref{claim:marked-disparition:induction} for \(k=N-1\).
\end{proof}

Put together, the previous lemmas prove that, given the right conditions, the execution converges deterministically. It only remains to show that these conditions happen often enough. These is formalized in the following lemma.
 
\begin{lemma}
    Consider an execution \(\execution\) and let \(t \ge N\). Let $\tau = 2^{N+2} + N$.  With probability at least $e^{-2^4} 4^{-N}$, the following happens:
    \begin{enumerate}[(i)]
        \item No marked train are emitted in the rounds $t, \dots, t+\tau$;
        \item 
        There exists a time $\tau' < 2^N + 2N$ such that, in time $t+\tau + \tau'$, a single leader emits a marked train;
        \item No other marked trains are emitted in rounds $[t+\tau, t+\tau + \tau' + 2D]$.
    \end{enumerate}
    \label{lem:probability-success}
\end{lemma}

\begin{proof}
From the description of the algorithm, and in particular of the  $\WagonUpdateLeader$ function in \cref{def:wagon-update-leader}, the trains are marked randomly by the leaders every $N$ rounds. In particular each leader, every $N$ rounds, creates a new train and marks it independently with probability $1/4^N$ (note that this is true only after \(t \ge N\), since the random variable \(v.\rand\) is initialized by the adversarial daemon, and not at random). For this reason, the events in $(i)$ is independent from the events in $(ii)$ and $(iii)$. 

Since there are at most $n$ leaders in the network, the event in $(i)$ happens with probability at least
\begin{align*}
  \left(1-\frac{1}{4^N}\right)^{n\tau} & \stackrel{(a)}{\geq} \exp\left(-\frac{2n\tau}{4^N}\right)
  \\ & \stackrel{(b)}{\geq} \exp \left(-\frac{2^3 n}{2^N}\right)
  \\ & \stackrel{(c)}{\geq}\exp\left(-2^3\right).
\end{align*}
where $(a)$ and $(c)$ follows since $N \geq \log n$ and for $n$ large enough, and $(b)$ follows since $\tau \leq 2^{N+2}+N$.

Then, after time $t+\tau$, we have no guarantees that there is a leader in the network. Let $\tau'$ be the first time a leader is in the network for $N$ consecutive rounds, and it is such that $v.F.\id = N$ (i.e. it will start emitting a new train). The leader needs to be in the network for $N$ rounds since, when it is created, it just send unmarked trains for the first $N$ rounds (see \cref{def:new-leader} for details). Note by the leader cannot be eliminated since there are no marked trains by \Cref{lem:marked-disparition}. By \Cref{lem:leader-creation}, we have that $\tau' < 2^N + 2N$. At time $t+\tau'$, the leader will start emitting a marked train with probability $1/4^N$. This ensures that the event $(ii)$ happens. In the same time, we have to show that no other leader emit marked trains in rounds $[t+\tau, t+\tau + \tau' +2D]$, and this will happen with probability at least
\begin{align*}
    \left(1-\frac{1}{4^N}\right)^{n(\tau' + 2D)} &\stackrel{(a)}{\geq } \exp \left(-\frac{2^{N}(2 \cdot 2^{N+1})}{4^N}\right) = \exp(-2^2).
    \end{align*}
where $(a)$ follows since $2D \leq 2n \leq 2^{N+1}$ and $n \leq 2^{N}$, and $\tau' \leq 2^{N+1}$.

Moreover, since the events are all independent, we have that $(i)$, $(ii)$ and $(iii)$ happen with probability at least $\exp(-2^4)4^{-N}$, which concludes the proof.
\end{proof}

We can finally prove the convergence of our algorithm.

\begin{lemma}
If $N \geq \log n$, our algorithm converges, with probability at least $1-1/n$, in $\cO(2^{3N} \log n)$ rounds.
\end{lemma}

\begin{proof}
Consider a time $t \geq N$. Assume the situation described in \Cref{lem:probability-success} happens at time \(t\). By \Cref{lem:marked-disparition}, there is no marked wagon at time \(t+\tau+\tau'-1\). Then, by \Cref{lem:leg-grow}, the system reaches a legitimate configuration at a time \(t' \le t + \tau+ \tau'+2D+1\) (indeed, for all \(v^*\in V\) we have \(\epsilon(v^*) \le D\)).
Thus, \Cref{lem:probability-success} implies that, with probability at least $e^{-2^4}4^{-N}$, we reach a legitimate configuration in the next $2^{N+3} + 3N + 2D + 1 \leq 2^{N+8}$ rounds.

Therefore, at each time $t$, we have an attempt to convergence that is successful with probability at least $e^{-2^4}4^{-N}$, where the attempts is taking $2^{N+8}$ rounds. We hence have that, after $ T = 4^N e^{2^4} \log n$ attempts, we reach convergence w.h.p., in particular we have that the first $T$ attempts are not successful with probability at most
\begin{align*}
   \left( 1-\frac{e^{-2^4}}{4^N}\right)^{T} \leq  \frac{1}{n}.
\end{align*}
Since each attempts is taking $2^{N+8}$ rounds, we have convergence in $4^N e^{2^4} \cdot 2^{N+8} \log n= \cO(2^{3N}\log n)$ rounds, proving the lemma.
\end{proof}

\section{Conclusions}
\label{sec:conclusion}
In this paper, we propose a self-stabilizing leader election algorithm that uses $\cO(\log \log n)$ bits of memory per node and stabilizes within a polynomial number of rounds. Our solution assumes a synchronous scheduler and relies on the common knowledge of a parameter $N = \Theta(\log n)$.

We cannot claim space optimality, since the currently known lower bounds on the memory required per node to solve leader election are $\Omega(\log \log n)$ bits in the deterministic setting and, in the probabilistic setting, $\Omega(\log \log n)$ for some values of $n$ under an unfair scheduler. It therefore remains open whether $\Omega(\log \log n)$ is also a lower bound under a synchronous
scheduler in anonymous networks, or whether algorithms using less memory can be designed.

Other open questions include whether algorithms using $\cO(\log \log n)$ bits of
memory per node can be constructed for asynchronous schedulers under any form
of fairness, and whether the assumption of common knowledge of $N$ can be
relaxed -- for example, by allowing nodes to possess differing approximate
estimates of $n$.

Finally, we believe that our information train technique could be applied to other related problems in the self-stabilizing setting while maintaining sublogarithmic memory usage. For instance, it may enable the design of deterministic leader election algorithms in ID-based networks, or be adapted to other problems, such as the spanning-tree construction or the clock synchronization problem.

\bibliographystyle{plain}
\bibliography{biblio}
\end{document}